\documentclass[11pt]{article} 

\usepackage[margin=1in]{geometry} 
\usepackage[utf8]{inputenc}
\usepackage[T1]{fontenc}
\usepackage{mathpazo} 

\usepackage{pifont}
\newcommand{\cmark}{\textcolor{Emerald}{\ding{51}}\xspace}
\newcommand{\xmark}{\textcolor{red}{\ding{55}}\xspace}

\usepackage{comment,cprotect,xspace}

\usepackage{booktabs,subcaption}

\usepackage{enumitem} 
\usepackage[dvipsnames]{xcolor} 
\definecolor{ptblue}{RGB}{15,76,129} 
\definecolor{ptemerald}{HTML}{009473} 

\usepackage[normalem]{ulem} 

\usepackage{tikz} 

\usepackage{amsmath,amsfonts,amssymb,amsthm,bbm,mathtools}

\let\OLDland\land
\renewcommand{\land}{\:\OLDland\:}
\let\OLDlor\lor
\renewcommand{\lor}{\:\OLDlor\:}
\let\OLDforall\forall
\renewcommand{\forall}{\;\OLDforall\:}
\let\OLDexists\exists
\renewcommand{\exists}{\;\OLDexists\,}

\DeclareMathOperator*{\argmax}{arg\,max}
\DeclareMathOperator*{\argmin}{arg\,min}
\newcommand{\indicator}[1]{\mathbbm{1}_{\left\{#1\right\}}\xspace}

\usepackage[ruled,linesnumbered,vlined]{algorithm2e}
\SetKwRepeat{Do}{do}{while} 

\SetCommentSty{mycommentfont}

\usepackage[square,sort]{natbib} 
\usepackage{hyperref}
\hypersetup{
linktocpage,
colorlinks=true,
citecolor=ptemerald, 
urlcolor=ptblue, 
linkcolor=Plum, 
}
\usepackage{cleveref} 
\usepackage{crossreftools} 
\theoremstyle{plain}
\newtheorem{theorem}{Theorem}[section]
\newtheorem{claim}[theorem]{Claim}

\newtheorem{lemma}[theorem]{Lemma}
\newtheorem{proposition}[theorem]{Proposition}

\theoremstyle{definition}
\newtheorem{definition}[theorem]{Definition}

\theoremstyle{remark}

\newcommand{\Real}{\mathbb{R}}
\newcommand{\E}{\mathbb{E}}

\newcommand{\cost}{\operatorname{cost}}
\newcommand{\approvalBallotOfAgent}[1]{A_{#1}} 
\newcommand{\BBOne}{BB1\xspace} 
\newcommand{\BFX}{BFx\xspace} 

\makeatletter 
\newcommand{\feasibleFracOutcomes}[1]{\if\relax\detokenize\expandafter{\@firstofone#1{}}\relax\mathcal{X}\else\mathcal{X}\left( #1 \right)\fi}
\makeatother

\newcommand{\EJR}{\text{EJR}\xspace}
\newcommand{\FJR}{\text{FJR}\xspace}

\newcommand{\MES}{\text{MES}\xspace}
\newcommand{\GCR}{\text{GCR}\xspace}

\newcommand{\budgetMES}{b}

\newcommand{\BWGCRforPB}{BW-GCR-PB\xspace}
\newcommand{\BWMESforPB}{BW-MES-PB\xspace}

\title{Fair Lotteries for Participatory Budgeting}
\author{Haris Aziz \qquad
Xinhang Lu \qquad
Mashbat Suzuki \qquad
Jeremy Vollen \qquad
Toby Walsh \bigskip\\
UNSW Sydney, Australia \medskip\\
\texttt{\{haris.aziz, xinhang.lu, mashbat.suzuki, j.vollen, t.walsh\}@unsw.edu.au}}
\date{}

\begin{document}
\maketitle

\begin{abstract}
In pursuit of participatory budgeting (PB) outcomes with broader fairness guarantees, we initiate the study of lotteries over discrete PB outcomes.
As the projects have heterogeneous costs, the amount spent may not be equal ex ante and ex post.
To address this, we develop a technique to bound the amount by which the ex-post spend differs from the ex-ante spend---the property is termed \emph{budget balanced up to one project (BB1)}.
With respect to fairness, we take a best-of-both-worlds perspective, seeking outcomes that are both ex-ante and ex-post fair.
Towards this goal, we initiate a study of ex-ante fairness properties in PB, including Individual Fair Share (IFS), Unanimous Fair Share (UFS) and their stronger variants, as well as Group Fair Share (GFS).
We show several incompatibility results between these ex-ante fairness notions and existing ex-post concepts based on justified representation.
One of our main contributions is a randomized algorithm which simultaneously satisfies ex-ante Strong UFS, ex-post full justified representation (FJR) and ex-post BB1 for PB with binary utilities.
\end{abstract}

\section{Introduction}

Participatory Budgeting (PB) is one of the exciting democratic paradigms that facilitates members of a community, municipality, or town to collectively make public project funding decisions.
Considering that PB generalizes classical voting and committee selection problems in social choice theory, it also poses interesting axiomatic and algorithmic research challenges~\citep{AzSh21,ReMa23}.
A major effort underway in computational social choice is to design meaningful axioms that capture elusive properties such as fairness and representation, and to design computationally efficient algorithms that satisfy such axioms.
There is also a movement to apply such rules in various countries (see, e.g., \url{https://equalshares.net}).

While existing algorithms for selecting discrete projects satisfy meaningful proportional representation properties, they do not provide any minimal representation guarantees to individual voters, even in an ex-ante sense.
For example, existing rules allow for the possibility that certain voters do not like any of the projects selected by the PB process.
Indeed, this is an inevitable reality for any deterministic algorithm.
One approach to address this issue is the use of randomization to achieve stronger fairness properties ex-ante.
This approach has been successful in various contexts, such as resource allocation~\citep{BoMo01a}, voting~\citep{BMS05a}, and committee voting~\citep{CJM+20a}.

In this paper, we initiate the study of randomization for PB.
We begin by studying the problem of implementation in PB: \textit{Given a marginal probability for each project, how can we compute a probability distribution (or lottery) over discrete outcomes that realizes these probabilities?}
While this question has already been answered in the PB special case of committee voting~\citep{CJM+20a}, we find that the PB setting with heterogeneous costs gives rise to significant obstacles on this front.
Since, unlike committee voting, we cannot guarantee that the total spend is equal ex-ante and ex-post, we develop a technique which bounds the amount by which the ex-post spend differs from the ex-ante spend.

Given our new implementation approach, we then return to the objective of strong ex-ante fairness properties in PB.
Our goal is to achieve ex-ante fairness while also ensuring that our lotteries only include ex-post fair outcomes, thus guaranteeing desirable fairness, both ex-ante and ex-post.
This approach is known as the \emph{best-of-both-worlds fairness} perspective, and has been employed in other adjacent contexts, such as fair division~\citep[e.g.,][]{AFS+23a} and committee voting~\citep{ALS+23}.
In this paper, we define a hierarchy of ex-ante fairness properties for PB and then consider a hierarchy of PB settings, giving algorithms which guarantee best-of-both-worlds fairness and/or incompatibility results for each setting considered.

\subsection{Our Results}

In \Cref{sec:PB:BOBW-budget}, we tackle the question of implementation in PB.
We first show that, unlike the committee voting setting, fractional outcomes cannot always be implemented by a lottery over integral outcomes using the same amount of budget.
Given this, we define a well-motivated axiom for lotteries --- \emph{budget balanced up to one (BB1)} --- which enforces a natural bound on the amount by which the total ex-post cost can differ from the cost of the fractional outcome the lottery implements.
We then demonstrate an approach which gives an implementation satisfying our axiom for any fractional outcome, and complement this result by showing that a lottery satisfying a natural ``up to any'' strengthening of our axiom may not always exist.

In \Cref{sec:PB:ex-ante}, we extend a hierarchy of well-studied properties capturing ex-ante fairness from the committee voting setting to our most general PB setting.
The ``fair share'' hierarchy of fairness axioms begins with the basic notion that each voter should receive at least a $1/n$ fraction of their optimal utility.
On the other hand, the ``strong fair share'' hierarchy starts with the stronger guarantee that each voter should be able to control their proportion of the budget.

We then turn to the goal of achieving best-of-both-worlds fairness in PB.
In \Cref{sec:PB:approval}, we investigate the special case of PB with binary utilities.
We show that our strongest ex-ante fairness notion cannot be guaranteed in tandem with any of our ex-post fairness notions (\Cref{sec:PB:approval:impossibility:GFS+JR}), notable since this is not the case in committee voting.
We then give a strong, positive result using an exponential-time algorithm (\Cref{sec:PB:approval:SUFS+FJR}) and a slightly weaker positive result using a polynomial-time algorithm (\Cref{sec:PB:approval:SUFS+EJR}).
Lastly, in \Cref{sec:PB:general}, we show that in the general model with cardinal utilities, ex-ante and ex-post fairness are not compatible, even for the weakest pair of axioms and even in the restricted case where projects are of unit cost. Our best-of-both-worlds fairness results are summarized in \Cref{fig:results_summary}.

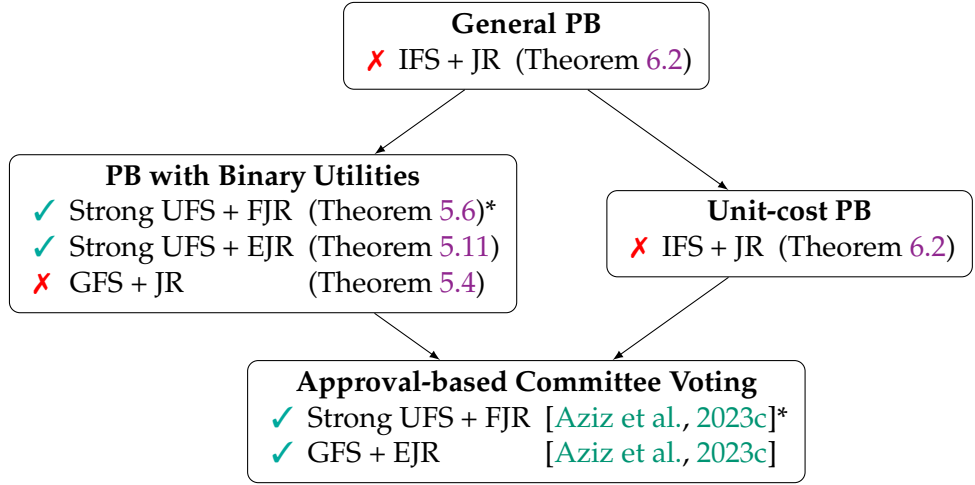
\begin{figure}
\centering
\begin{tikzpicture}[>=latex]
\node[draw, rounded corners, align=center] at (0,0) (general) {
\textbf{General PB} \\
\addtolength{\tabcolsep}{-0.3em}
\begin{tabular}{lll}
\xmark & IFS + JR & (\Cref{thm:PB:unit-cost:impossibility:IFS+JR})\\
\end{tabular}
};

\node[draw, rounded corners, align=center] at (-3.5,-2.5) (approval_pb) {
\textbf{PB with Binary Utilities} \\
\addtolength{\tabcolsep}{-0.3em}
\begin{tabular}{lll}
\cmark & Strong UFS + FJR & (\Cref{thm:PB:approval:sUFS+FJR})*\\
\cmark & Strong UFS + EJR & (\Cref{thm:PB:approval:sUFS+EJR})\\
\xmark & GFS + JR & (\Cref{thm:PB:approval:impossibility:GFS+JR})\\
\end{tabular}
};

\node[draw, rounded corners, align=center] at (3.5,-2.5) (unitcost) {
\textbf{Unit-cost PB} \\
\addtolength{\tabcolsep}{-0.3em}
\begin{tabular}{lll}
\xmark & IFS + JR & (\Cref{thm:PB:unit-cost:impossibility:IFS+JR})\\
\end{tabular}
};

\node[draw, rounded corners, align=center] at (0,-5) (abc) {
\textbf{Approval-based Committee Voting} \\
\addtolength{\tabcolsep}{-0.3em}
\begin{tabular}{lll}
\cmark & Strong UFS + FJR & \citep{ALS+23}* \\
\cmark & GFS + EJR & \citep{ALS+23} \\
\end{tabular}
};

\draw[->]
(general) edge (approval_pb) (approval_pb) edge (abc);
\draw[->]
(general) edge (unitcost) (unitcost) edge (abc);
\end{tikzpicture}
\caption{Summary of best-of-both-worlds fairness results in PB and special cases.
Arrows point from generalizations to special cases.
Compatibility results are represented by \cmark and impossibility results by \xmark.
(*) denotes exponential-time results.}
\label{fig:results_summary}
\end{figure}

\subsection{Related Work}

There is a fast growing body of work investigating proportionally representative outcomes in indivisible PB \citep[e.g.,][]{ALT18a,PPS21a,LCG22,MSWW22,BFL+23a,ReMa23}.
Since PB generalizes committee voting \citep{LaSk23a}, work on proportionality in PB often extends axioms and algorithms from the literature on proportional representation in committee voting~\citep{ABC+17,SEL+17,EFI+22,BrPe23a}.
It is to this literature that we are adding the tool of randomization.

Best-of-both–worlds fairness has been examined in several papers in the context of resource allocation~\citep{BEF22,AFS+23a,AGM23a,HSV23a,FMNP23}.
A couple of the earlier papers in this line of work are by \citet{AFS+23a} and \citet{Aziz19a}.
In recent work, \citet{ALS+23} initiated a best-of-both-worlds fairness perspective for the social choice setting of approval-based committee voting.
\citet{SuVo24} later proposed a stronger ex-ante fairness notion and provided an improved best-of-both-worlds fairness result.
We extend this approach to multiple PB settings, all of which generalize approval-based committee voting.

Indeed, we are the first to study lotteries in PB.
In the single-winner voting setting, \citet{BMS05a} computed ``mixtures'' of public outcomes which satisfy natural axioms seeking to give groups of agents their fair share, the same set of axioms which have inspired the ex-ante axioms in this work.

Lottery implementation techniques  have been studied in other social choice settings.
For example, the dependent randomized rounding technique of \citet{GKPS06} has been employed to compute randomized outcomes which implement desirable fractional outcomes in various settings such as fair division~\citet{AkNi20a} and committee selection~\citep{CJM+20a,MSWW22}.
In the context of apportionment, \citet{ALM+19a} and \citet{GPP22a} have created new rounding techniques to facilitate randomization when distributing legislative seats.

There are also several papers which study divisible PB, wherein projects can be funded fractionally.
\citet{FainGoMu16} showed that the Nash solution in this setting is in the \emph{core}, a property which captures proportional representation.
In this paper and most work in this area, projects do not have fixed costs, and thus any distribution of budget amongst projects is feasible, which contrasts significantly with our setting.
\citet{GKS+19a} incorporated project costs in the divisible PB setting and study strategic concerns.
However, their outcomes still allow for fractional project funding.
In this paper, we investigate the question of how such an outcome can be converted into a lottery over outcomes in which funding decisions are binary.

\section{Preliminaries}
\label{sec:PB:prelim}

For any positive integer~$t \in \mathbb{N}$, let $[t] \coloneqq \{1, 2, \dots, t\}$.
A participatory budgeting (PB) \emph{instance} is represented as a tuple $I = \langle N, C, \cost, B, (u_i)_{i \in [n]} \rangle$, where:
\begin{itemize}
\item $N = [n]$ and $C = [m]$ are the set of \emph{voters} and \emph{projects}, respectively.

\item $\cost \colon C \to \Real_{\geq 0}$ is the \emph{cost function}, associating each project~$c \in C$ with its cost that needs to be paid if~$c$ is selected.
For any subset of projects~$T \subseteq C$, denote by $\cost(T) \coloneqq \sum_{c \in T} \cost(c)$ the total cost of~$T$.

We say projects have \emph{unit costs} if $\cost(c) = 1$ for all~$c \in C$, and refer to the setting as \emph{unit-cost PB}.

\item $B \in \Real_{\geq 0}$ is the \emph{budget limit}.
We assume without loss of generality that $\cost(c) \leq B \forall c \in C$ and $\cost(C) \geq B$.

\item For each~$i \in N$, \emph{utility function} $u_i \colon C \to \mathbb{R}_{\geq 0}$ expresses how voter~$i$ values each project.
This most general formulation is referred to as \emph{general PB}.
We call the set of projects for which voter~$i$ has non-zero utility their \emph{approval set} and denote it by $\approvalBallotOfAgent{i}$.

We say voters have \emph{binary utilities} if $u_i(c) \in \{0, 1\}$ for all~$i \in N, c \in C$, and refer to the setting as \emph{PB with binary utilities}.
\end{itemize}

\paragraph{Integral Outcomes}
An \emph{integral outcome} (or simply \emph{outcome}) is a set of projects $W \subseteq C$, and it is said to be \emph{feasible} if $\cost(W) \leq B$.
We assume \emph{additive} utilities, meaning that given a subset of projects~$T \subseteq C$, $u_i(T) \coloneqq \sum_{c \in T} u_i(c)$.

\paragraph{Fractional Outcomes}
A \emph{fractional outcome} is an $m$-dimensional vector $\vec{p} \in [0, 1]^m$, where the component $p_c \in [0, 1]$ represents the fraction of project~$c$ funded.
Given an integral outcome~$W$, for notational convenience, let $\vec{1}_{W} \in \{0, 1\}^m$ be the binary vector whose $j$'th component is~$1$ if and only if $j \in W$.
Let $\cost(\vec{p}) \coloneqq \sum_{c \in C} p_c \cdot \cost(c)$.
A fractional outcome~$\vec{p}$ is said to be \emph{feasible} if $\cost(\vec{p}) = B$.
Given budget~$B'$, denote by~$\feasibleFracOutcomes{B'}$ the space of feasible fractional outcomes.
For simplicity, we use~$\feasibleFracOutcomes{}$ to denote the space of feasible fractional outcomes given budget~$B$.
Given a fractional outcome~$\vec{p}$, voter~$i$'s utility is denoted by $u_i(\vec{p}) \coloneqq \sum_{c \in C} p_c \cdot u_i(c)$.

\paragraph{Lotteries and Implementation}
A \emph{lottery} is a probability distribution over integral outcomes.
Formally, a lottery is specified by a set of~$s \in \mathbb{N}$ tuples $\{(\lambda_j, W_j)\}_{j \in [s]}$, where $\lambda_j \in [0, 1]$, $\sum_j \lambda_j = 1$, and for every $j \in [s]$, the integral outcome $W_j \subseteq C$ is selected with probability~$\lambda_j$.
A lottery $\{(\lambda_j, W_j)\}_{j \in [s]}$ is called an \emph{implementation} of (or, interchangeably, \emph{implements}) a fractional outcome~$\vec{p}$ if $\vec{p} = \sum_{j \in [s]} \lambda_j \cdot \vec{1}_{W_j}$.
In this paper, we only consider lotteries which implement feasible fractional outcomes.
We say a lottery satisfies a property \emph{ex ante} (resp., ex post) if the fractional outcome it implements (resp., every integral outcome in its support) satisfies the property.

This paper concerns the problem of designing (randomized) PB rules (or, interchangeably, algorithms) that simultaneously achieve desirable properties both ex ante and ex post.
Our algorithms do not explicitly output the desired lotteries (which in principle can be exponential in size), but instead sample integral outcomes from them.

\section{Implementing Fractional Outcomes in PB}
\label{sec:PB:BOBW-budget}

In this section, we study how to implement a fractional outcome in the context of participatory budgeting.
Decomposing a fractional outcome into a distribution over integral outcomes introduces novel challenges in the presence of costs.

Recall that implementing a fractional outcome $\vec{p}$ entails computing a probability distribution over integral outcomes $\Delta = \{(\lambda_j, W_j)\}_{j \in [s]}$ that realizes marginal probability $p_k$ for each project $k\in C$.
As a result, any implementation of a feasible fractional outcome has the property that $\E_{W \sim \Delta}[\cost(W)] = B$.
It is now easy to see that unless each integral outcome in the support of the lottery has cost equal to~$B$ (not possible in general), there must exist an integral outcome in the support of the lottery that exceeds the budget.

The aforementioned issue raises the natural question of whether it is possible to implement a fractional outcome while bounding the ex-post budget violations.
This is especially important in participatory budgeting since if the ex-post budget constraint is exceedingly violated, such an outcome is unlikely to be implemented in practice. We thus formalize an axiom which guarantees that an integral outcome is approximately within budget.

\begin{definition}[\BBOne]
An integral outcome~$W$ is said to be \emph{budget balanced up to one project (\BBOne)} if either
\begin{itemize}
\item $\cost(W) \leq B$ and there exists some project~$c \in C \setminus W$ such that $\cost(W \cup \{c\}) \geq B$, or
\item $\cost(W) \geq B$ and there exists some project~$c \in W$ such that $\cost(W \setminus \{c\}) \leq B$.
\end{itemize}
\end{definition}

We now show, perhaps surprisingly, that \emph{any} feasible fractional outcome~$\vec{p}$ can be implemented by a lottery over integral outcomes, each of which satisfies \BBOne.

\begin{theorem}
\label{thm:PB:Gandhi-et-al}
For any feasible fractional outcome~$\vec{p}$, there exists a random process running in polynomial time, that defines random variables $P_i \in \{0, 1\}$ for all~$i \in C$ such that the following properties hold:
\begin{enumerate}[label=\textbf{(P\arabic*)}]
\item\label{enum:marginal-distribution}
$\E[P_i] = p_i$ for each~$i \in C$;

\item\label{enum:degree-preservation}
Random integral committee $W = \{i \in C \mid P_i =1\}$ satisfies \BBOne with probability~$1$.
\end{enumerate}
\end{theorem}

\begin{proof}
The proof of \Cref{thm:PB:Gandhi-et-al} follows from applying the dependent rounding technique introduced by \citet{GKPS06}.

Given a fractional outcome $\vec{p} = (p_1, \dots, p_m)$ with $p_i \in [0, 1]$, we probabilistically modify each~$p_i$ to a random variable $P_i \in \{0, 1\}$ such that the random variables satisfy the properties \ref{enum:marginal-distribution} and \ref{enum:degree-preservation}.

We now describe the algorithm.
Let $\vec{q}^{\,0} = \vec{p}$.
We iteratively and randomly modify $\vec{q}^{\,0}$ in rounds.
Denote $\vec{q}^{\,t} = (q_1^t, q_2^t, \dots, q_m^t)$ as the values at round~$t$.
In each round, we update the values of at most two indices while keeping the values of all other indices constant.
Let $F^t = \{i \in C \mid  q_i^t \in (0,1)\}$ be the set of indices that are fractional in round~$t$.
The update rule depends on the cardinality of~$F^t$.

If $|F^t| \geq 2$, we arbitrarily select two indices $i, j \in F^t$ and run the following randomized update rule:
\[
(q^{t+1}_i, q^{t+1}_j) =
\begin{cases}
\left( q^t_i + \alpha, q^t_j -\frac{\cost(i)}{\cost(j)} \cdot \alpha \right) \; \text{ w.p. } \frac{\beta}{\alpha + \beta} \\
\left( q^t_i - \beta, q^t_j + \frac{\cost(i)}{\cost(j)} \cdot \beta \right) \; \text{ w.p. } \frac{\alpha}{\alpha + \beta}
\end{cases}
\]
where
\[
\alpha = \min \{\gamma > 0 \mid q^t_i + \gamma = 1 \text{ or } q^t_j - \frac{\cost(i)}{\cost(j)} \cdot \gamma = 0\},
\]
and
\[
\beta = \min \{\gamma > 0 \mid q^t_i - \gamma = 0 \text{ or } q^t_j + \frac{\cost(i)}{\cost(j)} \cdot \gamma = 1\}.
\]
For all other indices $\ell \in C \setminus \{i, j\}$, we set $q^{t+1}_\ell = q^t_\ell$.

If $|F^t| = 1$, we select the fractional index~$\ell \in F^t$ and set $q^{t+1}_\ell = 1$ with probability~$q^t_\ell$, and $q^{t+1}_\ell = 0$ with probability $1 - q^t_\ell$.

Finally, when no fractional indices exist, meaning $|F^t| = 0$ and hence $q^t_i \in \{0, 1\}$ for each~$i\in C$, we terminate the algorithm and set $P_i = q^t_i$ for all~$i \in C$.

The next two observations immediately follow from the algorithm's description.

\smallskip
\noindent\textbf{Observation 1)}
After each round, at least one index with a fractional value becomes integral i.e., $|F^{t+1}| < |F^{t}|$.

\smallskip
\noindent\textbf{Observation 2)}
Once a fractional index becomes integral, its values do not change.

Note that by the first observation and since $|F^0| \leq |C|$, we see that the number of rounds is at most~$|C|$.
As a result, we see that the random process does indeed run in polynomial time.

In what follows, we show the two properties and begin with the first property.

\begin{proof}[Proof of \ref{enum:marginal-distribution}]
For each~$i \in C$, let $P^t_i$ be the random variable denoting the value of~$q^t_i$.
We show that in each round, $\E[P^{t+1}_i] = \E[P^t_i]$ for each~$i \in C$.
The assertion trivially holds true for indices that were not updated, thus we focus on indices which were updated.

Suppose $|F^t| \geq 2$, and let $i, j$ be the indices which were chosen in the update rule. Then for all~$\zeta \in [0, 1]$,
\begin{align*}
\E[P^{t+1}_i \mid P^t_i = \zeta] &= \frac{\beta (\zeta + \alpha)}{\alpha + \beta} +  \frac{\alpha (\zeta - \beta)}{\alpha + \beta} = \zeta; \\
\E[P^{t+1}_j \mid P^t_j = \zeta] &= \left( \zeta - \frac{\cost(i)}{\cost(j)} \alpha \right) \frac{\beta}{\alpha + \beta} + \left( \zeta + \frac{\cost(i)}{\cost(j)} \beta \right) \frac{\alpha}{\alpha + \beta} = \zeta.
\end{align*}
Let $\Omega$ be the set of all possible values of $q^{t}_i$, then we see that
\begin{align*}
\E[P^{t+1}_i] = \sum_{\zeta \in \Omega} \E[P^{t+1}_i \mid P^t_i = \zeta] \cdot \Pr[P^t_i = \zeta] = \sum_{\zeta \in \Omega} \zeta \cdot \Pr[P^t_i = \zeta] = \E[P^t_i].
\end{align*}
An analogous argument can be used to show $\E[P^{t+1}_j] = \E[P^t_j]$.

Consider now the case in which the update rule is applied when $|F^t| = 1$.
Let~$\ell$ be the fractional index which was updated.
Then for all~$\zeta \in [0, 1]$,
\[
\E[P^{t+1}_\ell \mid P^t_\ell = \zeta] = 1 \cdot \zeta + 0 \cdot (1 - \zeta) = \zeta.
\]
Thus, we get that $\E[P^{t+1}_\ell] = \E[P^t_\ell]$ by a similar argument as seen in the previous paragraph.

It follows that in each round $\E[P^{t+1}_i] = \E[P^t_i]$ for all~$i \in C$.
Let~$t_f$ be the round in which the algorithm terminates, recursively applying the identity we see that $\E[P_i] = \E[P^{t_f}_i] = \E[P^0_i] = p_i$ for each~$i \in C$.
\end{proof}

We now proceed to show the second property.

\begin{proof}[Proof of \ref{enum:degree-preservation}]
We need to show that the random integral outcome $W = \{i \in C \mid P_i = 1\}$ satisfies \BBOne with probability~$1$.
We first show that whenever the update rule is applied with $|F^t| \geq 2$, it holds that $\sum_{i \in C} \cost(i) \cdot q_i^{t+1} = \sum_{i \in C} \cost(i) \cdot q_i^t$.
Suppose indices~$i, j$ were selected to be updated in round~$t$.
Then, regardless of the realization of the randomized update rule, the following holds:
\begin{align*}
\cost(i) \cdot q^{t+1}_i + \cost(j) \cdot q^{t+1}_j
= &\begin{cases}
&\cost(i) \cdot (q^t_i + \alpha) + \cost(j) \cdot \left( q^t_j - \frac{\cost(i)}{\cost(j)} \cdot \alpha \right) \\
&\qquad\qquad\qquad \text{or} \\
&\cost(i) \cdot (q^t_i - \beta) + \cost(j) \cdot \left( q^t_j + \frac{\cost(i)}{\cost(j)} \cdot \beta \right)
\end{cases} \\
= &\cost(i) \cdot q^t_i + \cost(j) \cdot q^t_j.
\end{align*}
As values of all other indices $C \setminus \{i, j\}$ were unchanged, we see that $\sum_{i \in C} \cost(i) \cdot q_i^{t+1} = \sum_{i \in C} \cost(i) \cdot q_i^t$.
Thus, noting that $\sum_{i \in C} \cost(i) \cdot q^0_i = \sum_{i \in C} \cost(i) \cdot p_i = B$, we have $\sum_{i \in C} \cost(i) \cdot q_i^{t+1} = B$ whenever the update rule is applied with $|F^t| \geq 2$.
Note that in all rounds, except possibly the last round, we have $|F^t| \geq 2$.
Hence we see that $\sum_{i \in C} \cost(i) \cdot q_i^{t_f - 1} = B$, where~$t_f$ is the round when the algorithm terminates.

In round~$t_f-1$, either the update rule is applied with $|F^{t_f-1}| \geq 2$ or $|F^{t_f-1}| = 1$.
If the update rule is applied with $|F^{t_f-1}| \geq 2$, then $\sum_{i \in C} \cost(i) \cdot q_i^{t_f} = \sum_{i \in C} \cost(i) \cdot q_i^{t_f-1} = B$ and hence $\sum_{i \in C} \cost(i) \cdot P_i = B$, meaning that the random committee $W = \{i \in C \mid P_i = 1\}$ is \BBOne.

Now suppose that $|F^{t_f-1}| = 1$ and let $\ell \in F^{t_f-1}$.
We may write,
\begin{align*}
B &= \sum_{i \in C} \cost(i) q_i^{t_f-1} = \sum_{i \in C \setminus \ell} \cost(i) q_i^{t_f-1} + \cost(\ell) q_\ell^{t_f-1} \\
&= \sum_{i \in C \setminus \ell} \cost(i) q_i^{t_f}+ \cost(\ell) q_\ell^{t_f-1} \\
&= \sum_{i \in C \setminus \ell} \cost(i) P_i + \cost(\ell) q_\ell^{t_f-1}.
\end{align*}
Here we used the fact that the update rule only changed the value of index~$\ell$ in round~$t_f-1$, and thus $P_i = q_i^{t_f} = q_i^{t_f-1}$ for all~$i \in C \setminus \ell$.
Recall $W=\{i\in C \ | \ P_i=1\}$.
If $q_\ell^{t_f}=1$, then $W$ satisfies \BBOne since $\cost(W) \geq B$ and removing~$\ell$ makes it under budget.
Similarly, if $q_\ell^{t_f} = 0$, then $\cost(W \cup \{\ell\}) \geq B \geq \cost(W)$.
Thus, we have that~$W$ satisfies \BBOne with probability one.
\end{proof}

We have now established \Cref{thm:PB:Gandhi-et-al}.
\end{proof}

Note that there is a lottery associated with the random process described in \Cref{thm:PB:Gandhi-et-al}, but we only return an integral outcome sampled from this underlying lottery as it may be exponential in size.
By~\ref{enum:marginal-distribution}, the underlying lottery implements~$\vec{p}$, and by~\ref{enum:degree-preservation}, it satisfies ex-post \BBOne.
We remark that \Cref{thm:PB:Gandhi-et-al} is a consequence of applying the dependent rounding scheme of \citet{GKPS06} to our setting.

One might wonder whether we can further strengthen the ex-post budget feasibility axiom \BBOne.
A natural strengthening is the following:

\begin{definition}[\BFX]
An integral outcome~$W$ is said to be \emph{budget feasible up to any project (\BFX)} if for all~$c \in W$, $\cost(W \setminus \{c\}) \leq B$.
\end{definition}

It is worth noting that \BFX is weaker than the natural ``up to any'' strengthening of \BBOne.
In particular, \BFX only bounds the amount an outcome can exceed the budget, and places no restriction on outcomes which are under budget.
We, however, show that not all fractional outcomes may be implemented by ex-post \BFX lotteries.

\begin{proposition}
\label{prop:BFx}
There exists some fractional outcome~$\vec{p}$ that cannot be implemented by a lottery that is ex-post \BFX.
\end{proposition}

\begin{proof}
    Consider an instance with budget~$B$ and three projects~$C = \{a,b,c\}$ such that $\cost(a) = \varepsilon$ and $\cost(b) = \cost(c) = \frac{B}{2} + \varepsilon$.
    Consider the fractional outcome~$\vec{p} = \left( 1, \frac{B - \varepsilon}{B + 2 \varepsilon}, \frac{B - \varepsilon}{B + 2 \varepsilon} \right)$.
    We now show that~$\vec{p}$ cannot be implemented by ex-post \BFX integral outcomes.
    First, since $p_1 = 1$, project~$a$ needs to be included in every integral outcome.
    Next, the integral outcome~$C=\{a,b,c\}$ is not \BFX because $\cost(C \setminus \{a\}) = B + 2 \varepsilon > B$.
    It means that the lottery can positive support only on integral outcomes~$\{a, b\},\{a, c\},\{a\},\{b\},\{c\}$. However, such a lottery cannot implement $\vec{p}$ as $\cost(\vec{p})=B$ and every integral outcome in the support of the lottery has cost strictly less than~$B$.
\end{proof}

\paragraph{Handling Hard Constraints}
We note that, in addition to being well-suited to scenarios in which budget constraints have some flexibility, the implementation techniques introduced in this section are also relevant to settings with hard ex-post budget constraints.
To see this, consider a problem wherein every ex-post outcome is restricted to having a cost of at most~$B$.
If we now apply \Cref{thm:PB:Gandhi-et-al} to any fractional outcome that spends $B' = B - \max_{g \in C} \cost(g)$, the resulting implementation has the property that every integral outcome in its support has cost at most~$B$.

\section{Ex-ante Fairness Concepts}
\label{sec:PB:ex-ante}

We now present ex-ante axioms for the PB setting inspired by the fair share axioms first introduced by \citet{BMS05a}.
The fair share axioms were recently extended by \citet{ALS+23} to the committee voting setting, the special case of our own model in which each project is of unit cost and voters have binary utilities.
In that work, \citet{ALS+23} highlighted two alternative interpretations of the idea behind fair share, and defined axiom hierarchies for each of these interpretations.
These two interpretations are given as follows:
\begin{enumerate}[label=(\alph*)]
\item\label{enum:fair-share}
Fair share: each voter is given~$1/n$ probability to choose their favourite outcome, or

\item\label{enum:strong-fair-share}
Strong fair share: each voter can select~$1/n$ of the outcome.
\end{enumerate}

In this section, we generalize the (strong) fair share axiom hierarchies to the general PB setting, with the intention of formulating axioms which (i) collapse to those defined by \citet{ALS+23} in approval-based committee voting, and (ii) reflect their respective interpretations as detailed above.
Each of the axioms given in this section provides lower bounds on utilities derived from fractional outcomes.
We note that these utilities can also be interpreted as expected utilities from implementations of these fractional outcomes.
In particular, if $\Delta = (\lambda_j, W_j)_{j \in [r]}$ is an implementation of a fractional outcome~$\vec{p}$, then $\E_{W \sim \Delta}[u_i(W)] = u_i(\vec{p})$.

The first axiom in our hierarchy, \emph{individual fair share (IFS)}, guarantees each agent receives utility which is at least a $1/n$-fraction of the utility they receive from their favourite fractional outcome.

\begin{definition}[IFS]
\label{def:PB:IFS}
A fractional outcome~$\vec{p}$ is said to provide \emph{IFS} if for each~$i \in N$,
\[
u_i(\vec{p}) \geq \frac{1}{n} \cdot \max_{\vec{t} \in \feasibleFracOutcomes{}} u_i(\vec{t}).
\]
\end{definition}

Note that the quantity expressed by the $\max$-operator is the utility-maximizing fractional outcome for the agent $i$, and hence it is immediately clear that \Cref{def:PB:IFS} follows interpretation~\ref{enum:fair-share}.
In general, this can be computed by selecting projects in order of descending utility per cost.
For binary utilities, this means selecting approved projects in order of ascending cost.
We can already observe that, in our setting, IFS seems quite a bit more demanding than its approval-based committee voting counterpart, in which a project/candidate can only take on a utility per cost value of one or zero. In contrast, in the PB setting, each voter-project pair could result in a unique utility per cost value.

\begin{definition}[Strong IFS]
A fractional outcome~$\vec{p}$ is said to provide \emph{Strong IFS} if for each~$i \in N$,
\[
u_i(\vec{p}) \geq \max_{\vec{t} \in \feasibleFracOutcomes{\frac{B}{n}}} u_i(\vec{t}).
\]
\end{definition}

For Strong IFS, keeping with interpretation \ref{enum:strong-fair-share} above, an agent's utility lower bound is given by the optimal utility they could achieve if given their proportion of the budget.
Next, we strengthen IFS to \emph{unanimous fair share (UFS)}, which strengthens the fair share utility guarantee for groups of voters with identical preferences.

\begin{definition}[UFS]
A fractional outcome~$\vec{p}$ is said to provide \emph{UFS} if for any~$S \subseteq N$ where $u_i = u_j$ for any~$i, j \in S$, then the following holds for each~$i \in S$:
\[
u_i(\vec{p}) \geq \frac{|S|}{n} \cdot \max_{\vec{t} \in \feasibleFracOutcomes{}} u_i(\vec{t}).
\]
\end{definition}

\begin{definition}[Strong UFS]
\label{def:PB:sUFS}
A fractional outcome~$\vec{p}$ is said to provide \emph{Strong UFS} if for any~$S \subseteq N$ where $u_i = u_j$ for any~$i, j \in S$, the following holds for each~$i \in S$:
\begin{equation}
\label{eq:PB:sUFS-guarantee}
u_i(\vec{p}) \geq \max_{\vec{t} \in \feasibleFracOutcomes{|S| \cdot \frac{B}{n}}} u_i(\vec{t}).
\end{equation}
\end{definition}

As its name suggests, Strong UFS implies UFS (and hence Strong IFS implies IFS).\footnote{To see this, let $\vec{q}=\argmax_{\vec{t} \in \feasibleFracOutcomes{}} u_i(\vec{t})$. Now simply note that $\frac{|S|}{n} \vec{q} \in \feasibleFracOutcomes{|S| \cdot \frac{B}{n}}$.}
While (Strong) UFS gives a utility guarantee to groups of voters with identical preferences, our next axiom --- \emph{group fair share (GFS)} --- extends a non-trivial representation guarantee to all groups of voters.

\begin{definition}[GFS]
\label{def:gfs}
A fractional outcome~$\vec{p}$ is said to provide \emph{GFS} if the following holds for any~$S \subseteq N$:
\[
\sum_{j \in C} \Large( p_j \cdot \max_{i \in S} u_{ij} \Large) \geq \frac{1}{n} \cdot \sum_{i \in S} \max_{\vec{t} \in \feasibleFracOutcomes{}} u_i(\vec{t}).
\]
\end{definition}

In committee voting, the LHS of GFS is simply the probability allocated to candidates in the union of the group of voters' approval sets.
Thus, while it is clear that our definition collapses to the GFS of \citet{ALS+23} in committee voting, ours is not the only formulation of the LHS of GFS which does so.
For example, instead of taking the maximum utility for each project~$j$ over all agents in~$S$, we could have instead taken the average or median (or minimum) utility over all agents in~$S$ with non-zero utility for project~$j$.
Of these options, our formulation results in the weakest definition.
Since, as we will see, this definition of GFS is not compatible with any of the ex-post fairness notions we consider in any PB setting, each of the results considering GFS in this paper would also hold for any stronger notions of GFS.

\paragraph{Fractional Random Dictator}
We now extend the well-known Random Dictator algorithm~\citep{BMS05a} to the computation of fractional PB outcomes.
The high-level idea of the algorithm is to compute the fractional outcome resulting from giving each agent $1/n$ probability to select their own favourite fractional outcome.
For an agent~$i \in N$, let~$X_i$ be the maximal set of projects which can be funded fully in order of maximum utility per cost and let~$g_i$ be the project with highest utility per cost in the remaining set of projects.
Also, denote the indicator function by~$\indicator{\cdot}$.
For each~$j \in C$, the \emph{Fractional Random Dictator} algorithm outputs the fractional outcome~$\vec{p}$ defined as follows:
\[
p_j = \frac{1}{n} \sum_{i \in N} \indicator{j \in X_i} + \indicator{j = g_i} \frac{B - \cost(X_i)}{\cost(j)}.
\]

\begin{theorem}
\label{thm:PB:general:RD-properties}
The Fractional Random Dictator algorithm computes an ex-ante GFS fractional outcome.
\end{theorem}

\begin{proof}
Let $\vec{p}$ be the fractional outcome returned by the Fractional Random Dictator mechanism. We first show that $\vec{p}$ is a feasible fractional outcome:
\begin{align*}
\sum_{j\in C} p_j\cdot \cost(j)
= &\frac{1}{n} \sum_{j\in C} \cost(j) \sum_{i\in N} \left[\indicator{j\in X_i} + \indicator{j=g_i} \frac{B-\cost(X_i)}{\cost(j)}\right] \\
= &\frac{1}{n} \sum_{i\in N}  \sum_{j\in C} \cost(j) \left[\indicator{j\in X_i} + \indicator{j=g_i} \frac{B-\cost(X_i)}{\cost(j)}\right] \\
= &\frac{1}{n} \sum_{i\in N}  \cost(X_i) + \cost(g_i) \frac{B-\cost(X_i)}{\cost(g_i)}
= B.
\end{align*}

Lastly, we show that $\vec{p}$ satisfies ex-ante GFS:
\begin{align*}
\sum_{j \in C} p_j \cdot \max_{i \in S} u_{ij}
= &\frac{1}{n} \sum_{j \in C} \max_{i \in S} u_{ij} \left[ \sum_{i\in N} \indicator{j\in X_i} + \indicator{j=g_i} \frac{B-\cost(X_i)}{\cost(j)} \right] \\
\geq &\frac{1}{n} \sum_{j \in C} \left[ \sum_{i\in S} u_{ij} \cdot \indicator{j\in X_i} + u_{ij} \cdot \indicator{j=g_i} \frac{B-\cost(X_i)}{\cost(j)} \right] \\
\geq &\frac{1}{n} \sum_{i\in S} \sum_{j \in C} u_{ij} \cdot \indicator{j\in X_i} + u_{ij} \cdot \indicator{j=g_i} \frac{B-\cost(X_i)}{\cost(j)}  \\
\geq &\frac{1}{n} \sum_{i\in S} \left[ \sum_{j \in X_i} u_{ij} + u_{ig_i} \frac{B-\cost(X_i)}{\cost(g_i)} \right]  \\
= &\frac{1}{n} \cdot \sum_{i \in S} \max_{\vec{t} \in \feasibleFracOutcomes{}} u_i(\vec{t}).
\end{align*}

The last inequality follows since the quantity in brackets is exactly equal to the utility agent $i$ receives from the fractional outcome they select in the Fractional Random Dictator mechanism, which is their optimal fractional outcome.
\end{proof}

\section{BoBW Fairness in PB with Binary Utilities}
\label{sec:PB:approval}

In this section, we consider the setting of PB with binary utilities, that is, the voters have binary utilities while the projects have arbitrary costs.
Our main focus is to investigate whether the ex-ante fair share notions defined in \Cref{sec:PB:ex-ante} can be achieved simultaneously with ex-post fairness properties like \emph{justified representation (JR)}, \emph{extended justified representation (EJR)} and \emph{full justified representation (FJR)}~\citep{PPS21a}.

\begin{definition}[$T$-cohesive group for PB with binary utilities]
\label{def:PB:approval:T-cohesive}
A group of voters~$S \subseteq N$ is said to be \emph{$T$-cohesive} for $T \subseteq C$ if $|S| \cdot \frac{B}{n} \geq \cost(T)$ and $T \subseteq \bigcap_{i \in S} \approvalBallotOfAgent{i}$.
\end{definition}

\begin{definition}[JR \& \EJR for PB with binary utilities]
An outcome~$W$ is said to satisfy
\begin{itemize}
\item \emph{justified representation (JR)} if for each~$j \in C$ and every $\{j\}$-cohesive group of voters~$S \subseteq N$, it holds that $u_i(W) = |\approvalBallotOfAgent{i} \cap W| \geq 1$ for some~$i \in S$; and

\item \emph{extended justified representation (EJR)} if for each~$T \subseteq C$ and every $T$-cohesive group of voters~$S \subseteq N$, it holds that $u_i(W) = |\approvalBallotOfAgent{i} \cap W| \geq |T|$ for some~$i \in S$.
\end{itemize}
\end{definition}

By definition, \EJR implies JR.
We next introduce a notion stronger than \EJR.

\begin{definition}[\FJR for PB with binary utilities]
Given a positive integer~$\beta$ and a set of projects~$T \subseteq C$, a group of voters~$S \subseteq N$ is said to be \emph{weakly $(\beta, T)$-cohesive} if $|S| \cdot \frac{B}{n} \geq \cost(T)$ and $|\approvalBallotOfAgent{i} \cap T| \geq \beta$ for all~$i \in S$.

An outcome~$W$ is said to satisfy \emph{full justified representation (FJR)} if for every weakly $(\beta, T)$-cohesive group of voters~$S \subseteq N$, it holds that $|\approvalBallotOfAgent{i} \cap W| \geq \beta$ for some~$i \in S$.
\end{definition}

The remainder of this section is organized as follows:
\begin{itemize}
\item In \Cref{sec:PB:approval:impossibility:GFS+JR}, we show that it is impossible to simultaneously achieve ex-ante GFS and ex-post JR.
\item In \Cref{sec:PB:approval:SUFS+FJR}, we show constructively that ex-ante Strong UFS and ex-post \FJR are compatible, though our randomized algorithm is not polynomial time.
\item In \Cref{sec:PB:approval:SUFS+EJR}, we devise a polynomial-time randomized algorithm which simultaneously achieves ex-ante Strong UFS and ex-post \EJR.
\end{itemize}

\subsection{Impossibility: Ex-ante GFS + Ex-post JR}
\label{sec:PB:approval:impossibility:GFS+JR}

Our first main result in PB with binary utilities states that it is impossible to simultaneously achieve ex-ante GFS and ex-post JR.
Note that in the more restricted setting with unit-cost projects (i.e., approval-based committee voting), \citet{ALS+23} showed that ex-ante GFS is compatible even with ex-post EJR.
Our impossibility result demonstrates a clear and strong separation between PB with binary utilities and approval-based committee voting.

\begin{theorem}
\label{thm:PB:approval:impossibility:GFS+JR}
In PB with binary utilities, ex-ante GFS and ex-post JR are incompatible.
\end{theorem}

\begin{proof}
Consider an instance with $n \geq 6$ and the following approval sets and project costs:
\begin{itemize}
\item each voter~$i \in N$ approves $\approvalBallotOfAgent{i} = \{g^*, a_i, b_i, c_i\}$ with $\cost(g^*) = \frac{B}{2}$ and $\cost(a_i) = \cost(b_i) = \cost(c_i) = \frac{B}{2} - \varepsilon$, where $\varepsilon < \frac{B}{2} - \frac{2B}{n}$;

\item note that $g^*$ is the common project approved by every voter, and for any pair of voters~$i \neq j$, $a_i, b_i, c_i \notin A_j$.
\end{itemize}

We establish the incompatibility using this instance by showing that \emph{any} feasible fractional outcome satisfying GFS cannot be implemented by \emph{any} lottery that is ex-post JR, even without imposing \BBOne.
Suppose, for the sake of contradiction, that $\{(\lambda_j, W_j)\}_{j \in [s]}$ is an ex-post JR lottery implementing GFS fractional outcome $\vec{p}$.

We first point out that some integral outcome in the lottery includes~$g^*$, and hence $p_{g^*} > 0$.

\begin{claim}
\label{claim:PB:approval:impossibility:GFS+JR}
There exists an outcome~$W_j$ such that $g^* \in W_j$.
\end{claim}

\begin{proof}[Proof of Claim]
Suppose for the sake of contradiction that every integral outcome does not contain~$g^*$.
Fix any outcome~$W_j$.
Consider the set of voters $N' = \{i \in N \mid \approvalBallotOfAgent{i} \cap W_j = \emptyset\}$.
Recall that for any pair of voters~$i \neq i'$, $a_i, b_i, c_i \notin \approvalBallotOfAgent{i'}$.
If $|W_j| < n/2$, then $|N'| \geq n/2$.
It means that~$W_j$ is not JR because every voter in the $\{g^*\}$-cohesive group~$N'$ gets zero utility.
Therefore, $|W_j| \geq n/2$.
Since $g^* \notin W_j$ and every other project has an identical cost of~$\frac{B}{2} - \varepsilon$, we have $\cost(W_j) \geq \frac{n}{2} \cdot \left( \frac{B}{2} - \varepsilon \right)$.
Moreover, as lottery~$\{(\lambda_j, W_j)\}_{j \in [s]}$ is an implementation of the feasible fractional outcome~$\vec{p}$, we have
\[
B = \sum_{j \in [s]} \lambda_j \cdot \cost(W_j) \geq \sum_{j \in [s]} \lambda_j \cdot \frac{n}{2} \cdot \left( \frac{B}{2} - \varepsilon \right),
\]
which implies $\sum_{j \in [s]} \lambda_j \leq \frac{2}{n} \cdot \frac{B}{B/2 - \varepsilon} < 1$, contradicting the assumption that $\{(\lambda_j, W_j)\}_{j \in [s]}$ is an implementation of~$\vec{p}$.
\end{proof}

Feasibility of~$\vec{p}$ means that $B = \sum_{c \in C} p_c \cdot \cost(c) = \sum_{c \in C \setminus \{g^*\}} p_c \cdot \left( \frac{B}{2} - \varepsilon \right) + p_{g^*} \cdot \frac{B}{2}$. Thus,
\[
\sum_{c \in C} p_c = \frac{B - B/2 \cdot p_{g^*}}{B/2 - \varepsilon} + p_{g^*} = \frac{B - \varepsilon \cdot p_{g^*}}{B/2 - \varepsilon}.
\]
Since $\vec{p}$ satisfies GFS with respect to voters~$N$, we thus have
\begin{align*}
\frac{B - \varepsilon \cdot p_{g^*}}{B/2 - \varepsilon}
= \sum_{c \in C} p_c
\geq \frac{1}{n} \cdot \sum_{i \in N} \max_{\vec{t} \in \feasibleFracOutcomes{}} u_i(\vec{t})
= \frac{1}{n} \cdot \sum_{i \in N} \frac{B}{B/2 - \varepsilon} = \frac{B}{B/2 - \varepsilon},
\end{align*}
a contradiction because~$p_{g^*} > 0$.
\end{proof}

As demonstrated by \citet{ALS+23} in the approval-based committee voting, there is no logical dependence between GFS and Strong UFS.
It is thus unclear whether ex-ante Strong UFS can be compatible with any ex-post fairness properties.
We answer the question in the affirmative below.

\subsection{Ex-ante Strong UFS + Ex-post FJR}
\label{sec:PB:approval:SUFS+FJR}

We now show that if we only focus on giving ex-ante fair share guarantees to \emph{unanimous} (instead of any) groups, ex-ante Strong UFS is compatible even with ex-post FJR.

\begin{theorem}
\label{thm:PB:approval:sUFS+FJR}
In PB with binary utilities, \Cref{alg:pb_bw_gcr} computes an integral outcome sampled from a lottery that is ex-ante Strong UFS, ex-post \BBOne and ex-post \FJR.
\end{theorem}

\subsubsection{The Algorithm: \BWGCRforPB}

\begin{algorithm}[t]
\caption{\BWGCRforPB: Strong UFS and FJR}
\label{alg:pb_bw_gcr}
\DontPrintSemicolon

\KwIn{Voters~$N = [n]$, projects~$C = [m]$, cost function~$\cost$, budget~$B$, and utilities~$(\vec{u}_i)_{i \in N}$.}

$W_\GCR \gets \GCR(N, C, \cost, B, (u_i)_{i \in N})$\; \label{alg:pb_bw_gcr:GCR}
$\vec{p} \gets \vec{1}_{W_\GCR}$\;

$\widetilde{N} \gets \emptyset$\;
$b_i \gets 0$ for all~$i \in N$\;

Let $\{N^1, \dots, N^\eta\}$ be the unanimous groups of~$N$.\;
\ForEach{$z \in [\eta]$}{
	\If{$|\approvalBallotOfAgent{N^z} \cap W_\GCR| = |G_{N^z}|$}{
		$\widetilde{N} \gets \widetilde{N} \cup N^z$\;
		$b_i \gets \frac{B}{n} - \frac{1}{|N^z|} \cdot \cost(G_{N^z})$ for all~$i \in N^z$\; \label{alg:pb_bw_gcr:b_i}
		Let voters~$N^z$ spend their total budget $|N^z| \cdot \frac{B}{n} - \cost(G_{N^z})$ on project~$c \in \approvalBallotOfAgent{N^z}$ with the smallest cost, provided the updated~$p_c \leq 1$.\;
	}
}

Increase~$\vec{p}$ arbitrarily such that for all~$c \in C$, $p_c \leq 1$ and $\cost(\vec{p}) = B$.\; \label{alg:pb_bw_gcr:feasibility}

Obtain an outcome~$W$ sampled from the lottery implementing~$\vec{p}$ by applying \Cref{thm:PB:Gandhi-et-al}.\;

\Return{$\vec{p}$ and~$W$}
\end{algorithm}

Our algorithm, whose pseudocode can be found in \Cref{alg:pb_bw_gcr}, starts by feeding the given PB instance into the \emph{Greedy Cohesive Rule (\GCR)} of \citet{PPS21a} and obtains an \FJR outcome.
More specifically, \GCR begins by making all voters as \emph{active} and initializing~$W = \emptyset$.
In each step, \GCR searches for a set of voters~$N' \subseteq N$ who are all active and a set of projects~$T \subseteq C \setminus W$ such that~$N'$ is weakly $(\beta, T)$-cohesive, breaking ties in favour of larger~$\beta$, next smaller~$\cost(T)$, and then larger~$|N'|$.
\GCR then includes projects~$T$ to~$W$ and labels voters~$N'$ as inactive.
If, at any step, no weakly $(\beta, T)$-cohesive group exists for any positive integer~$\beta$, then \GCR returns~$W$ and terminates.
Denote by~$r$ the number of steps \GCR executes before terminating.
For each~$j \in [r]$, we refer to~$\beta_j$, $T_j$ and $N_j$ as the values of~$\beta$, $T$ and~$N'$ for the weakly cohesive group selected in the $j$-th step of \GCR.
Denote by~$W_\GCR \coloneqq \bigcup_{j \in [r]} T_j$ and initialize~$\vec{p}$ as $\vec{1}_{W_\GCR}$.

\Cref{alg:pb_bw_gcr} next loops over unanimous groups and set budgets for the voters.
We first introduce additional notation for each unanimous group.
Fix any unanimous group~$S \subseteq N$.
We denote by~$\approvalBallotOfAgent{S}$ the approval set of the unanimous group~$S$ (i.e., for all~$i \in S$, $\approvalBallotOfAgent{i} = \approvalBallotOfAgent{S}$).
Let us rename the projects in~$A_S$ in non-decreasing order of cost with arbitrary tie-breaking, i.e., $\cost(g_1) \leq \cost(g_2) \leq \dots \leq \cost(g_{|A_S|})$.
Denote by $G_S \coloneqq \{g_1, g_2, \dots, g_{\kappa_S}\}$ the maximal set of projects such that $\cost(G_S) \leq |S| \cdot \frac{B}{n}$.
Put differently, if $A_S \setminus G_S \neq \emptyset$, $\cost(G_S \cup \{g_{\kappa_S+1}\}) > |S| \cdot \frac{B}{n}$.

For ease of expression, let $\{N^1, N^2, \dots, N^\eta\}$ be the partition of the \emph{(maximal) unanimous groups} of voters~$N$, i.e., for each~$z \in [\eta]$, voters~$N^z$ are unanimous and for any~$i \in N^z$ and~$i' \in N^{z'}$ with~$z \neq z'$, $\approvalBallotOfAgent{i} \neq \approvalBallotOfAgent{i'}$.
Fix any~$z \in [\eta]$.
If voter~$i \in N^z$ gets utility exactly~$|G_{N^z}|$ from~$W_\GCR$, i.e., the \verb|if|-condition holds, we set budget~$b_i \coloneqq \frac{B}{n} - \frac{1}{|N^z|} \cdot \cost(G_{N^z})$.
The unanimous group of voters~$N^z$ then spend their total budget of $|N^z| \cdot \frac{B}{n} - \cost(G_{N^z})$ on project~$c \in \approvalBallotOfAgent{N^z}$ with the smallest cost, provided the updated~$p_c \leq 1$.
We will show shortly that the budget set-up is valid (\Cref{lem:PB:approval:BW-GCR-PB:valid-budgets}) and guarantee that each unanimous group satisfies Strong UFS (\Cref{lem:PB:approval:BW-GCR-PB:sUFS}).
\Cref{alg:pb_bw_gcr:feasibility} then increases~$\vec{p}$ in an arbitrary way so that~$\vec{p}$ is feasible, that is, for all~$c \in C$, $p_c \leq 1$ and $\cost(\vec{p}) = B$.

Finally, given the feasible fractional outcome~$\vec{p}$, we apply the randomized rounding scheme (\Cref{thm:PB:Gandhi-et-al}) and sample an outcome from the lottery implementing~$\vec{p}$.

\subsubsection{The Analysis of \BWGCRforPB}

We begin by stating two key lemmas in the proof of \Cref{thm:PB:approval:sUFS+FJR}.
First, we prove that the total budgets we give the voters in \cref{alg:pb_bw_gcr:b_i} is upper bounded by the leftover budget limit after selecting~$W_\GCR$.

\begin{lemma}
\label{lem:PB:approval:BW-GCR-PB:valid-budgets}
$\sum_{i \in N} b_i \leq B - \cost(W_\GCR)$.
\end{lemma}

\begin{proof}
For ease of exposition, in this proof, we \emph{re-order} the~$r$ weakly cohesive groups encountered by \GCR in \cref{alg:pb_bw_gcr:GCR}.
Let $r' \in \{0, 1, 2, \dots, r\}$ be an index such that for all~$j \in [r']$, $N_j \cap \widetilde{N} \neq \emptyset$, i.e., there exists a unanimous group in~$N_j$ such that the \verb|if|-condition holds.
For each~$j \in [r]$, let $\{N_j^1, N_j^2, \dots, N_j^{\eta_j}\}$ be the partition of the (maximal) unanimous groups of voters~$N_j$.
We also assume without loss of generality that the first~$\eta'_j \in \{0, 1, 2, \dots, \eta_j\}$ unanimous groups are the ones such that the \verb|if|-condition holds.
Note that $\eta'_j \geq 1$ for all~$j \in [r']$.
Denote by~$N_\GCR \coloneqq \bigcup_{j \in [r]} N_j$ the set of inactive voters due to \GCR.
Note that for all~$i \in N$, $b_i \leq \frac{B}{n}$.
We will also make use of the following claim:

\begin{claim}
\label{claim:PB:approval:cost-comparison}
$\forall j \in [r'], z \in [\eta'_j], \cost(T_j) \leq \cost(G_{N_j^z})$.
\end{claim}

\begin{proof}[Proof of Claim]
For ease of expression, let $T_{[j-1]} \coloneqq \bigcup_{t = 1}^{j-1} T_t$ denote the set of projects added by \GCR in the first~$j-1$ steps and $G'_{N_j^z} \coloneqq G_{N_j^z} \setminus T_{[j-1]}$.
Suppose for the sake of contradiction that $\cost(G_{N_j^z}) < \cost(T_j)$.
Recall that the unanimous group~$N_j^z$ is $G_{N_j^z}$-cohesive.
It follows that $N_j^z$ is $G'_{N_j^z}$-cohesive because $|N_j^z| \cdot \frac{B}{n} \geq \cost(G_{N_j^z}) \geq \cost(G'_{N_j^z})$ and $G'_{N_j^z} \subseteq \approvalBallotOfAgent{N_j^z}$.

We first show that $\beta_j = |G'_{N_j^z}|$.
If $\beta_j < |G'_{N_j^z}|$, then the $j$-th step of \GCR would have added projects~$G'_{N_j^z}$ (instead of~$T_j$) because~$N_j^z$ is $G'_{N_j^z}$-cohesive and \GCR breaks ties in favor of larger~$\beta$.
If $\beta_j > |G'_{N_j^z}|$, then $\beta_j \geq |G'_{N_j^z}| + 1$.
Recall that $j \in [r']$ and $z \in [\eta'_j]$, meaning $|\approvalBallotOfAgent{N_j^z} \cap W_\GCR| = |G_{N_j^z}|$.
We, however, have
\begin{align*}
|G_{N_j^z}|
&= |\approvalBallotOfAgent{N_j^z} \cap W_\GCR| \\
&\geq |\approvalBallotOfAgent{N_j^z} \cap T_{[j]}|
= |\approvalBallotOfAgent{N_j^z} \cap T_{[j-1]}| + |\approvalBallotOfAgent{N_j^z} \cap T_j| \\
&= |\approvalBallotOfAgent{N_j^z} \cap T_{[j-1]}| + \beta_j \\
&\geq |\approvalBallotOfAgent{N_j^z} \cap T_{[j-1]}| + |G'_{N_j^z}| + 1 \\
&\geq |G_{N_j^z} \cap T_{[j-1]}| + |G_{N_j^z} \setminus T_{[j-1]}| + 1 \\
&= |G_{N_j^z}| + 1,
\end{align*}
a contradiction.

However, if $\beta_j = |G'_{N_j^z}|$, then the $j$-th step of \GCR would have added projects~$G'_{N_j^z}$.
This is because $N_j^z$ is $G'_{N_j^z}$-cohesive and conditioning on the same~$\beta$ value, \GCR breaks ties in favor of~$G'_{N_j^z}$ due to $\cost(G'_{N_j^z}) \leq \cost(G_{N_j^z}) < \cost(T_j)$.
\end{proof}

We are now ready to establish the statement of the lemma:
\allowdisplaybreaks
\begin{align*}
\sum_{i \in N} b_i
&= \sum_{i \in N_\GCR} b_i + \sum_{i \in N \setminus N_\GCR} b_i \\
&= \sum_{j = 1}^{r'} \sum_{z = 1}^{\eta'_j} \left( |N_j^z| \cdot \frac{B}{n} - \cost(G_{N_j^z}) \right) + \sum_{i \in N \setminus N_\GCR} b_i \\
&\leq \sum_{j = 1}^{r'} \Bigg( |N_j| \cdot \frac{B}{n} - \sum_{z = 1}^{\eta'_j} \cost(G_{N_j^z}) \Bigg) + \sum_{i \in N \setminus N_\GCR} b_i \\
&\leq \sum_{j = 1}^{r'} \left( |N_j| \cdot \frac{B}{n} - \cost(T_j) \right) + \sum_{i \in N \setminus N_\GCR} b_i \\
&\leq \sum_{j = 1}^{r} \left( |N_j| \cdot \frac{B}{n} - \cost(T_j) \right) + \sum_{i \in N \setminus N_\GCR} b_i \\
&\leq |N_\GCR| \cdot \frac{B}{n} - \cost(W_\GCR) + |N \setminus N_\GCR| \cdot \frac{B}{n} \\
&= B - \cost(W_\GCR),
\end{align*}
where the fourth transition is due to \Cref{claim:PB:approval:cost-comparison} and $\eta'_j \geq 1$, and the fifth transition is due to weak cohesiveness.
\end{proof}

Our next result establishes that~$\vec{p}$ satisfies Strong UFS.

\begin{lemma}
\label{lem:PB:approval:BW-GCR-PB:sUFS}
\Cref{alg:pb_bw_gcr} outputs a fractional outcome~$\vec{p}$ that satisfies Strong UFS.
\end{lemma}

\begin{proof}
Let us first establish connections between unanimous groups considered by Strong UFS and cohesive groups considered by \EJR.\footnote{Since \FJR implies \EJR, our discussion is carried over to weakly cohesive groups considered by \FJR.}
Recall that $G_S \coloneqq \{g_1, g_2, \dots, g_{\kappa_S}\}$ is the maximal subset of projects approved by unanimous group $S$ (in non-decreasing order of cost) such that $\cost(G_S) \leq |S| \cdot \frac{B}{n}$.
Since $|S| \cdot \frac{B}{n} \geq \cost(G_S)$ and $G_S \subseteq A_S$, we observe that the unanimous group~$S$ is in fact $G_S$-cohesive.

It follows that given an \EJR (or \FJR) outcome~$W$, for \emph{all}~$i \in S$, $|\approvalBallotOfAgent{i} \cap W| \geq |G_S|$.
We thus conclude:

\begin{claim}
\label{claim:PB:approval:EJR-utility-guarantee}
Given an instance of PB with binary utilities, fix any unanimous group of voters~$S \subseteq N$ and any \EJR (or \FJR) outcome~$W$, then, for all~$i \in S$, $|\approvalBallotOfAgent{i} \cap W| \geq |G_S|$.
\end{claim}

We now provide an alternative description of \Cref{eq:PB:sUFS-guarantee} in the definition of Strong UFS.
According to the RHS of \Cref{eq:PB:sUFS-guarantee}, the group~$S$ is endowed with a budget of $|S| \cdot \frac{B}{n}$ to select a fractional outcome.
An optimal fractional outcome can be achieved by fully funding~$G_S$ and next funding a $\delta_S$ fraction of project~$g_{\kappa_S+1}$, where $\delta_S \coloneqq \frac{|S| \cdot \frac{B}{n} - \cost(G_S)}{\cost(g_{\kappa_S+1})} < 1$.
Thus, in this case, we can rewrite the RHS of \Cref{eq:PB:sUFS-guarantee}:
\begin{equation}
\label{eq:PB:sUFS-alternative}
\textstyle
\max_{\vec{t} \in \feasibleFracOutcomes{|S| \cdot \frac{B}{n}}} u_i(\vec{t}) = |G_S| + \delta_S.
\end{equation}

We are now prepared to show that each unanimous group is satisfied with respect to Strong UFS.
Fix any~$z \in [\eta]$ and any voter~$i \in N^z$.
Since $W_\GCR$ satisfies \FJR, by \Cref{claim:PB:approval:EJR-utility-guarantee}, $|\approvalBallotOfAgent{i} \cap W_\GCR| \geq |G_{N^z}|$.
If $\approvalBallotOfAgent{N^z} = G_{N^z}$, then $u_i(\vec{p}) \geq |\approvalBallotOfAgent{i} \cap W_\GCR| \geq |A_{N^z}|$, meaning that voter~$i$ is satisfied with respect to Strong UFS.
We thus assume from now on $\approvalBallotOfAgent{N^z} \setminus G_{N^z} \neq \emptyset$.
Put differently, the project~$g_{\kappa_{N^z} + 1}$ is well-defined.
We will use the RHS of \Cref{eq:PB:sUFS-alternative} as the target utility to reason about Strong UFS.
Specifically, we will show that $u_i(\vec{p}) \geq |G_{N^z}| + \delta_{N^z}$.
We will mainly distinguish cases between $|\approvalBallotOfAgent{i} \cap W_\GCR| \geq |G_{N^z}| + 1$ and $|\approvalBallotOfAgent{i} \cap W_\GCR| = |G_{N^z}|$.

If $|\approvalBallotOfAgent{i} \cap W_\GCR| \geq |G_{N^z}| + 1$, Strong UFS is satisfied as
\begin{align*}
u_i(\vec{p})
\geq |\approvalBallotOfAgent{i} \cap W_\GCR|
\geq |G_{N^z}| + 1
> |G_{N^z}| + \delta_{N^z}
= \max_{\vec{t} \in \feasibleFracOutcomes{|N^z| \cdot \frac{B}{n}}} u_i(\vec{t}).
\end{align*}

We now move on to the case where $|\approvalBallotOfAgent{i} \cap W_\GCR| = |G_{N^z}|$.
If $A_{N^z} \subseteq W_\GCR$, clearly, Strong UFS is already satisfied.
We thus assume $A_{N^z} \setminus W_\GCR \neq \emptyset$.
It means that there must exist a project~$g \in \{g_1, g_2, \dots, g_{\kappa_{N^z}}, g_{\kappa_{N^z} + 1}\}$ such that $g \notin W_\GCR$ and $\cost(g) \leq g_{\kappa_{N^z} + 1}$.
\Cref{alg:pb_bw_gcr:b_i} of \Cref{alg:pb_bw_gcr} gives voters~$N^z$ a total budget of $|N^z| \cdot \frac{B}{n} - \cost(G_{N^z})$ to spend on project~$g$.
It may be that project~$g$ has already been funded to some extent in previous step(s), but this only helps voter~$N^z$ accumulate higher utilities.
It follows easily that
\begin{align*}
u_i(\vec{p})
&\geq |\approvalBallotOfAgent{i} \cap W_\GCR| + \frac{|N^z| \cdot \frac{B}{n} - \cost(G_{N^z})}{\cost(g)} \\
&\geq |G_{N^z}| + \frac{|N^z| \cdot \frac{B}{n} - \cost(G_{N^z})}{\cost(g_{\kappa_{N^z} + 1})}
= |G_{N^z}| + \delta_{N^z}
=\max_{\vec{t} \in \feasibleFracOutcomes{|N^z| \cdot \frac{B}{n}}} u_i(\vec{t}).
\qedhere
\end{align*}
\end{proof}

Using these statements, we now prove \Cref{thm:PB:approval:sUFS+FJR}.

\begin{proof}[Proof of \Cref{thm:PB:approval:sUFS+FJR}]
We first show that the fractional outcome~$\vec{p}$ returned by \Cref{alg:pb_bw_gcr} is feasible, i.e., for all~$c \in C$, $p_c \leq 1$ and $\cost(\vec{p}) = B$.
First, by the design of \Cref{alg:pb_bw_gcr}, we maintain $p_c \leq 1$ for all~$c \in C$ at any step.
Next, we show that $\cost(\vec{p}) = B$.
By \Cref{lem:PB:approval:BW-GCR-PB:valid-budgets}, we have $\cost(W_\GCR) + \sum_{i \in N} b_i \leq B$, meaning that $\cost(\vec{p}) \leq B$ before \cref{alg:pb_bw_gcr:feasibility} executes.
Finally, according to how we increase~$\vec{p}$ in \cref{alg:pb_bw_gcr:feasibility} and our assumption that $\cost(C) \geq B$, we conclude that the final fractional outcome~$\vec{p}$ is feasible.

Next, by \Cref{lem:PB:approval:BW-GCR-PB:sUFS}, the fractional outcome~$\vec{p}$ returned by \Cref{alg:pb_bw_gcr} satisfies Strong UFS.

Finally, by \Cref{thm:PB:Gandhi-et-al}, the lottery which implements~$\vec{p}$ thus satisfies ex-ante Strong UFS as well as the sampled outcome~$W$ satisfies ex-post \BBOne.
Ex-post \FJR follows from \citet{PPS21a} ($W_\GCR$ is \FJR) and from \Cref{thm:PB:Gandhi-et-al} that $W_\GCR$ is included in the integral outcome sampled.
\end{proof}

\subsection{Ex-ante Strong UFS + Ex-post EJR (in Polynomial Time)}
\label{sec:PB:approval:SUFS+EJR}

Despite providing strong ex-ante and ex-post fairness guarantees, \BWGCRforPB does not run in polynomial time.
We present here a polynomial-time algorithm that is ex-ante Strong UFS, at the cost of weakening ex-post fairness guarantee to \EJR.

\begin{theorem}
\label{thm:PB:approval:sUFS+EJR}
In PB with binary utilities, \Cref{alg:pb_bw_mes} computes an integral outcome sampled from a lottery that is ex-ante Strong UFS, ex-post \BBOne, and ex-post EJR in polynomial time.
\end{theorem}

\begin{algorithm}[t]
\caption{\BWMESforPB: Strong UFS and EJR}
\label{alg:pb_bw_mes}
\DontPrintSemicolon

\KwIn{Voters~$N = [n]$, projects~$C = [m]$, budget~$B$, cost function~$\cost$, and utilities~$(\vec{u}_i)_{i \in N}$.}

$W_\MES \gets \MES(N, C, \cost, B, (u_i)_{i \in N})$\; \label{alg:pb_bw_mes:mes}
$\vec{p} \gets \vec{1}_{W_\MES}$\;
Let $y_{ij}$ for each~$i \in N$ and~$j \in W_\MES$ be the amount voter~$i$ spent on project~$j$ during \MES.\;
$\budgetMES_i \gets \frac{B}{n} - \sum_{j \in W_\MES} y_{ij}$ for all~$i \in N$, which is the \emph{remaining} budget of voter~$i$ after \MES

$N' \gets \{i \in N \mid \approvalBallotOfAgent{i} \setminus W_\MES \neq \emptyset\}$\;
\ForEach{$i \in N'$}{
	Let $\kappa_i \in \argmin_{c \in \approvalBallotOfAgent{i} \setminus W_\MES} \cost(c)$\; \label{alg:pb_bw_mes:pick-favorite-project}
	$y_{i \kappa_i} \gets \budgetMES_i$\;  \label{alg:pb_bw_mes:fund-favorite-project}
}
\ForEach{$i \in N \setminus N'$}{
	Voter~$i$ spends $\budgetMES_i$ arbitrarily provided $\sum_{i \in N} y_{ij} \leq \cost(j)$ for all~$j \in C$. \label{alg:pb_bw_mes:fund-arbitrarily}
}
\lForEach{$j \in C$}{
	$p_j \gets \frac{\sum_{i \in N} y_{ij}}{\cost(j)}$. \label{alg:pb_bw_mes:construct-fractional}
}

Obtain an outcome~$W$ sampled from the lottery implementing~$\vec{p}$ by applying \Cref{thm:PB:Gandhi-et-al}. \label{alg:pb_bw_mes:sample}

\Return{$\vec{p}$ and~$W$} \label{alg:pb_bw_mes:implementation}
\end{algorithm}

At a high level, our algorithm \BWMESforPB (\Cref{alg:pb_bw_mes}) gives each voter an initial budget of~$B/n$ and uses the \emph{Method of Equal Shares (\MES)} of \citet{PPS21a} as a subroutine to obtain an \EJR outcome~$W_\MES$.
We now formally define the Method of Equal Shares (\MES) and its necessary components.

\begin{definition}[\MES]
Each voter is initially given a budget of~$B/n$.
We start with~$W = \emptyset$ and sequentially add projects to~$W$.
For each selected project~$j \in W$, we write~$y_{ij}$ for the amount that voter~$i$ pays for~$j$; we require that $\sum_{i \in N} y_{ij} = \cost(j)$.
We write $\budgetMES_i = B/n - \sum_{j \in W} y_{ij} \geq 0$ for the amount of budget voter~$i$ has left.
For $\rho \geq 0$, we say that a project~$j \notin W$ is \emph{$\rho$-affordable} if
\[
\sum_{i \in N} \min (\budgetMES_i, u_i(j) \cdot \rho) = \cost(j).
\]
If no project is $\rho$-affordable for any $\rho$, \MES terminates and returns~$W$.
Otherwise, it selects the project~$j \in C \setminus W$ that is $\rho$-affordable for minimum $\rho$.
Payments are given by $y_{ij} = \min (\budgetMES_i, u_i(j) \cdot \rho)$.
\end{definition}

A key step in the proof of \Cref{thm:PB:approval:sUFS+EJR} is to show that for each unanimous group~$N^z \subseteq N$, the remaining budget of the group~$\sum_{i \in N^z} \left( \frac{B}{n} - \sum_{c \in W_\MES} y_{i c} \right)$ is at least $|N^z| \cdot \frac{B}{n} - \cost(G_{N^z})$.
As a result, the group together can use their remaining budget to fund the project with the smallest cost and be satisfied with respect to Strong UFS.
We now prove the theorem.

\begin{proof}[Proof of \Cref{thm:PB:approval:sUFS+EJR}]
First, we show that the fractional outcome~$\vec{p}$ returned by \Cref{alg:pb_bw_mes} is feasible, i.e., $\cost(\vec{p}) = B$.
By the design of \Cref{alg:pb_bw_mes}, at any step, we maintain $p_c \leq 1$ for all~$c \in C$.
Next, since each voter starts with a budget of $B/n$ and spends the entirety of their budget, by the construction of $\vec{p}$ in \cref{alg:pb_bw_mes:construct-fractional} we have that
\begin{align*}
\sum_{j \in C} p_j \cdot \cost(j)
= \sum_{j \in C} \frac{\sum_{i \in N} y_{ij}}{\cost(j)} \cdot \cost(j)
= \sum_{i \in N} \sum_{j \in C} y_{ij} = \sum_{i \in N} B/n = B.
\end{align*}

Next, we show that the fractional outcome~$\vec{p}$ satisfies Strong UFS.
The proof idea is similar to that of \Cref{lem:PB:approval:BW-GCR-PB:sUFS}.
Fix any~$z \in \eta$ and any voter~$i \in N^z$.
If $\approvalBallotOfAgent{i} \subseteq W_\MES$, then voter~$i$ already gets the highest possible utility and thus is satisfied with respect to Strong UFS.
We hence assume from now on that $\approvalBallotOfAgent{i} \setminus W_\MES \neq \emptyset$.

Since~$W_\MES$ satisfies \EJR, by \Cref{claim:PB:approval:EJR-utility-guarantee}, $|\approvalBallotOfAgent{i} \cap W_\MES| \geq |G_{N^z}|$.
If $\approvalBallotOfAgent{N^z} = G_{N^z}$, then
\[
u_i(\vec{p}) \geq |\approvalBallotOfAgent{i} \cap W_\MES| \geq |G_{N^z}| = |\approvalBallotOfAgent{N^z}|,
\]
implying that Strong UFS is satisfied.
We thus focus on the case where~$\approvalBallotOfAgent{N^z} \setminus G_{N^z} \neq \emptyset$.
In other words, the project~$g_{\kappa_{N^z} + 1}$ is well-defined.
From \Cref{eq:PB:sUFS-alternative}, it suffices for us to show $u_i(\vec{p}) \geq |G_{N^z}| + \delta_{N^z}$.
If $|\approvalBallotOfAgent{i} \cap W_\MES| \geq |G_{N^z}| + 1$, clearly, Strong UFS is satisfied.
We now consider the case where $|\approvalBallotOfAgent{i} \cap W_\MES| = |G_{N^z}|$.

It can be observed from the definition of \MES that for any unanimous group~$S \subseteq N$ and for any pair of voters~$i, j \in S$, $b_i = b_j$ at any step; moreover, for any project~$c \in C$, $y_{i c} = y_{j c}$.
We show below that
\begin{equation}
\label{eq:PB:approval:BW-MES-PB:payment-comparison}
|N^z| \cdot \sum_{c \in \approvalBallotOfAgent{N^z} \cap W_\MES} y_{i c} \leq \cost(G_{N^z}).
\end{equation}
Put differently, the payments the unanimous group~$N^z$ made during \MES is at most the amount needed to buy~$G_{N^z}$.
Recall that $G_{N^z} = \{g_1, g_2, \dots, g_{\kappa_{N^z}}\}$ is the best set of projects voters~$N^z$ can afford as a group.
Let $\{h_1, h_2, \dots, h_t\} \subseteq \approvalBallotOfAgent{i} \cap W_\MES$ denote the first~$t$ projects added by \MES from $\approvalBallotOfAgent{i} \cap W_\MES$.
More specifically, we show that for each~$t = 1, 2, \dots, |G_{N^z}|$,
\[
|N^z| \cdot y_{i h_t} \leq \cost(g_t).
\]
Suppose for the sake of contradiction that $|N^z| \cdot y_{i h_t} > \cost(g_t)$.
By the definition of being $\rho$-affordable, at the step where \MES includes project~$h_t$, the $\rho$-value is at least~$y_{i h_t}$.
Note that at the moment, there still exists some project~$g \in \{g_1, g_2, \dots, g_t\}$ available to be funded, and, $\cost(g) \leq \cost(g_t)$.
As a result, project~$g$ is $\rho$-affordable with $\rho$-value upper bounded by $\frac{\cost(g_t)}{|N^z|}$, which is less than~$y_{i h_t}$, a contradiction.

Let $h \in \argmin_{c \in \approvalBallotOfAgent{i} \setminus W_\MES} \cost(c)$.
Since $|\approvalBallotOfAgent{i} \cap W_\MES| = |G_{N^z}|$, we have $h \in \{g_1, g_2, \dots, g_{\kappa_{N^z}}, g_{\kappa_{N^z} + 1}\}$ and $\cost(h) \leq \cost(g_{\kappa_{N^z} + 1})$.
As a result, we are able to show that Strong UFS is satisfied:
\allowdisplaybreaks
\begin{align*}
\sum_{c \in \approvalBallotOfAgent{i}} p_c
&= \sum_{c \in \approvalBallotOfAgent{i}} \frac{\sum_{j \in N} y_{j c}}{\cost(c)} \\
&= \sum_{c \in \approvalBallotOfAgent{i} \cap W_\MES} \frac{\sum_{j \in N} y_{j c}}{\cost(c)} + \sum_{c \in \approvalBallotOfAgent{i} \setminus W_\MES} \frac{\sum_{j \in N} y_{j c}}{\cost(c)} \\
&= |\approvalBallotOfAgent{i} \cap W_\MES| + \sum_{c \in \approvalBallotOfAgent{i} \setminus W_\MES} \frac{\sum_{j \in N} y_{j c}}{\cost(c)} \\
&\geq |\approvalBallotOfAgent{i} \cap W_\MES| + \frac{\sum_{j \in N} y_{j h}}{\cost(h)} \\
&\geq |\approvalBallotOfAgent{i} \cap W_\MES| + \frac{\sum_{j \in N^z} y_{j h}}{\cost(h)} \\
&\geq |\approvalBallotOfAgent{i} \cap W_\MES| + \frac{|N^z| \cdot y_{i h}}{\cost(h)} \\
&= |\approvalBallotOfAgent{i} \cap W_\MES| + \frac{|N^z| \cdot \left( \frac{B}{n} - \sum_{c \in \approvalBallotOfAgent{i} \cap W_\MES} y_{i c} \right)}{\cost(h)} \\
&\geq |\approvalBallotOfAgent{i} \cap W_\MES| + \frac{|N^z| \cdot \frac{B}{n} - \cost(G_{N^z})}{\cost(h)} \quad (\because \text{\Cref{eq:PB:approval:BW-MES-PB:payment-comparison}}) \\
&\geq |\approvalBallotOfAgent{i} \cap W_\MES| + \frac{|N^z| \cdot \frac{B}{n} - \cost(G_{N^z})}{\cost(g_{\kappa_{N^z} + 1})} \\
&= |G_{N^z}| + \delta_{N^z} = \max_{\vec{t} \in \feasibleFracOutcomes{|N^z| \cdot \frac{B}{n}}} u_i(\vec{t}).
\end{align*}

Finally, by \Cref{thm:PB:Gandhi-et-al}, the lottery which implements~$\vec{p}$ satisfies ex-ante Strong UFS and the sampled outcome~$W$ satisfies ex-post \BBOne.
Ex-post \EJR follows from \citet{PPS21a} ($W_\MES$ is \EJR) and from \Cref{thm:PB:Gandhi-et-al} that $W_\MES$ is included in the integral outcome sampled.
\end{proof}

\section{BoBW Fairness in General PB}
\label{sec:PB:general}

We now move on to the setting of general PB, in which we show a strong impossibility that ex-ante IFS and ex-post JR are not compatible, even in the unit-cost PB setting.
This impossibility is striking as, even in the general setting, the much stronger properties of ex-ante GFS and ex-post FJR are independently achievable via Fractional Random Dictator (\Cref{thm:PB:general:RD-properties}) and \GCR~\citep{PPS21a}, respectively.

We first define the notion of \emph{justified representation} \citep{PPS21a}.

\begin{definition}[JR]
\label{def:PB:general:JR}
In the general PB setting, a group of voters $S$ is \emph{$(\alpha, T)$-cohesive}, where $\alpha:C\rightarrow[0,1]$ and $T\subseteq C$, if $|S|/n\geq \cost(T)/B$ and if $u_{ij} \geq \alpha(j)$ for all $i\in S$ and $j\in T$.
An integral outcome~$W$ satisfies JR, if for each $\alpha:C\rightarrow[0,1]$, $j\in C$, and each $(\alpha, \{j\})$-cohesive group of agents, there exists an agent $i\in S$ such that $u_i(\vec{1}_W) \geq \alpha(j)$.
\end{definition}

We show that ex-ante IFS and ex-post JR are incompatible in the unit-cost setting with cardinal utilities.
Intuitively, in situations where voters have high utilities for distinct projects, the outcomes that guarantee the highest expected utility may not include a project which gives every voter non-zero utility.

\begin{theorem}
\label{thm:PB:unit-cost:impossibility:IFS+JR}
In unit-cost PB, ex-ante IFS and ex-post JR are incompatible.
\end{theorem}

\begin{proof}
Consider any instance of unit-cost PB with $n \geq 4$, $m = 2n + 1$, and in which $k = 2$.
    Suppose the project set is composed of (i) one project $c$, for which $u_{ic} = 1$ for all $i\in N$, and additionally (ii) a set of two projects $G_i$ for each agent $i\in N$, such that agent $i$ gets value $H$ from each project in $G_i$ and every other agent $i'\neq i$ gets zero utility from each project in $G_i$.
    Note that $N$ constitutes a $(1,c)$-cohesive group.

    We begin by showing that any integral outcome which satisfies ex-post JR must contain $c$.
    Suppose, for a contradiction, that there exists some integral committee $W$ which satisfies ex-post JR and which does not contain $c$.
    Since $c\not\in W$, there are at least $n-2$ agents who receive zero utility from $W$. Denote this group of agents $S$.
    Note that $n\geq 4 \implies n-2\geq n/2$, and each agent in $S$ receives one utility from $c$. Thus, $S$ are $(1,c)$-cohesive.
    However, each agent in $S$ receive zero utility from $W$ which contradicts that $W$ satisfies ex-post JR.

    Thus, any randomized committee which satisfies ex-post JR for our instance must be of the form $\lbrace (\lambda_1, \{c, w_1\}), \ldots, (\lambda_r, \{c, w_r\})$ where, for each $j\in[r]$, $w_j\in H_i$ for some $i\in N$. Denote by $\vec{p}$ any fractional committee which such a randomized committee implements.
    If we sum the LHS of the IFS guarantees of all agents in N, we have that
    \begin{align*}
        \sum_{i\in N} u_i(\vec{p}) = \sum_{j\in [r]} \lambda_j \sum_{i\in N} (1 + u_{iw_j})
     = \sum_{j\in [r]} \lambda_j (n + H)
     = n + H.
    \end{align*}

    Thus, for some agent $i\in N$, it holds that $u_i(\vec{p}) \leq 1 + H/n.$ Finally, see that $\frac{1}{n} \max_{\vec{t}\in \mathcal{X}} u_i(\vec{t}) = 2H$.
    Together, these observations show us that for any instance in which $H>n$, there will be at least one which agent who does not receive their IFS guarantee.
\end{proof}

\section{Conclusion}

In this paper, we initiated the study of PB lotteries and used this approach to study best-of-both-worlds fairness in PB.
We provided a complete set of results for two natural restrictions of PB with cardinal utilities.\footnote{Another common way to measure each voter's satisfaction regarding a PB outcome is to use the total cost of this voter's approved projects; this setting is referred to as \emph{PB with cost utilities}.
We defer our best-of-both-worlds fairness results for this setting to \Cref{sec:PB:cost-utilities}.}
Specifically, we gave algorithms which compute a lottery that guarantees each voter certain expected utility while maintaining the strongest indivisible PB fairness notions ex post.
While we focused on fairness, it is an interesting future direction to seek best-of-both-worlds results for other desiderata, such as efficiency.

\section*{Acknowledgements}

This work was partially supported by the NSF-CSIRO grant on ``Fair Sequential Collective Decision-Making'' and the ARC Laureate Project FL200100204 on ``Trustworthy AI''.

\bibliographystyle{plainnat}
\bibliography{bibliography}

\clearpage
\appendix

\section{BoBW Fairness in PB with Cost Utilities}
\label{sec:PB:cost-utilities}

In this section, we consider the model restriction in which voters have \emph{cost utilities}, that is, $u_i(c) \in \{0, \cost(c)\}$ for all $i \in N, c \in C$.
We refer to this setting as \emph{PB with cost utilities}, which generalizes approval-based committee voting.
We point out that approval sets~$A_i$ are still well-defined in this setting as the set of projects~$c$ for which $u_i(c) \neq 0$.
Thus, \Cref{alg:pb_bw_mes} remains well-defined in PB with cost utilities.
\MES itself is defined for the general setting with additive utilities \citep{PPS21a}.
However, whereas in that setting, it only attains the fairness guarantee of \emph{EJR up to one project}, it is able to achieve the considerably stronger property of \emph{EJR up to any project (EJR-x)} under cost utilities \citep{BFL+23a}, which we will define now.

\begin{definition}[EJR-x for PB with cost utilities]
An outcome~$W$ is said to satisfy \emph{EJR up to any project (EJR-x)} if for each~$T \subseteq C$ and every $T$-cohesive group of voters~$S \subseteq N$, there is a voter~$i \in S$ such that $u_i(W \cup \{c\}) > u_i(T)$ for all~$c \in T\setminus W$.
\end{definition}

We will now show that, under cost utilities, \Cref{alg:pb_bw_mes} attains GFS and Strong UFS.

\begin{lemma}
\label{lem:cost_gfs}
In PB with cost utilities, \Cref{alg:pb_bw_mes} computes an integral outcome sampled from a lottery that is ex-ante GFS.
\end{lemma}

\begin{proof}
Recall that $y_{ij}$ is used to denote the amount of budget spent by voter $i$ on project $j$ over the course of \Cref{alg:pb_bw_mes}.
We begin our proof by showing that, for each~$i \in N$, it holds that:
\begin{equation}
\label{eq:claim}
\sum_{j\in A_i} y_{ij} \geq \frac1n \min \{ B, \cost(A_i)\}.
\end{equation}

First consider voter~$i \in N'$.
By the definition of \MES and the specification of \cref{alg:pb_bw_mes:fund-favorite-project}, voter~$i$ only spends on projects which she approves.
Since voter~$i$ spends her entire budget, we have that
\[
\sum_{j\in A_i} y_{ij} = \frac{B}n \geq \frac1n  \min \{ B, \cost(A_i)\}.
\]

Now suppose instead that~$i \not\in N'$ and thus $A_i\subseteq W_{\MES}$.
Fix~$j \in A_i$.
Project~$j$ was added in the \MES phase and thus $y_{ij} = \min (b_i, u_i(j) \cdot \rho_j)$ where~$b_i$ here denotes the amount of budget~$i$ had remaining at the step in which~$j$ was selected, and~$\rho_j$ is the value of~$\rho$ which determined the payments for project~$j$.
If $b_i \leq u_i(j) \cdot \rho_j$, then voter~$i$ spends his entire budget during the course of \MES and \Cref{eq:claim} follows.
If not, then
\[
y_{ij} = u_i(j) \cdot \rho_j = \cost(j) \cdot \rho_j \geq \frac1n \cost(j).
\]
To see why the first inequality holds, assume for a contradiction that $\rho_j<\frac1n$.
Then,
\[
\sum_{i\in N} \min(b_i, u_{ij}\cdot \rho_j) \leq \sum_{i\in N_j} \cost(j)\cdot \rho_j < \cost(j),
\]
and thus~$j$ is not $\rho_j$-affordable.
This leads to a contradiction. Thus, we have shown that each voter $i$ either spends their entire budget during \MES or $y_{ij} \geq \frac1n \cost(j)$ for all $j\in A_i$. \Cref{eq:claim} follows.

Fix any~$S \subseteq N$.
We can now show that~$\vec{p}$, the fractional outcome returned by \Cref{alg:pb_bw_mes}, satisfies GFS:
\begin{align*}
\sum_{j \in C} \Large( p_j \cdot \max_{i \in S} u_{ij} \Large) &= \sum_{j\in \bigcup_{i\in S} A_i} p_j\cdot \cost(j) \\
&= \sum_{j\in \bigcup_{i\in S} A_i} \frac{\sum_{i\in N} y_{ij}}{\cost(j)} \cdot \cost(j) \\
&\geq \sum_{i\in S} \sum_{j\in A_i} y_{ij} \\
&\geq \sum_{i\in S} \frac1n \min \{ B, \cost(A_i)\} \\
&\geq \frac{1}{n} \sum_{i \in S} \max_{\vec{t} \in \feasibleFracOutcomes{}} u_i(\vec{t}).
\qedhere
\end{align*}
\end{proof}

\begin{lemma}
\label{lem:cost_sufs}
    In PB with cost utilities, \Cref{alg:pb_bw_mes} computes an integral outcome sampled from a lottery that is ex-ante Strong UFS.
\end{lemma}

\begin{proof}
    Consider an arbitrary unanimous group $S\subseteq N$. Denote $A=A_i \forall i\in S$ as the approval set of each voter in $S$.
    If $A\subseteq W_\MES$, then Strong UFS follows immediately since each voter in $S$ is achieving their optimal utility.
    If not, then each voter expends the remainder of their budget in \cref{alg:pb_bw_mes:fund-favorite-project}, and thus for each $i\in S$, we have that
    \begin{align*}
         u_i(\vec{p}) = \sum_{c\in A} p_c \cdot \cost(c) = \sum_{c\in A} \sum_{i\in N} y_{ic} \geq \sum_{i\in S} \sum_{c\in A} y_{ic} \ = \  |S| \frac{B}n \ \geq \ \max_{\vec{t} \in \feasibleFracOutcomes{|S| \cdot \frac{B}{n}}} u_i(\vec{t}).
     \end{align*}
     The last inequality follows because, due to cost utilities, each unit of budget spent can bring at most one unit of utility to the voter.
\end{proof}

Combining the ex-ante fairness properties shown by the previous two lemmata with the known ex-post fairness of \MES under cost utilities, we obtain the following best-of-both-worlds fairness result.

\begin{theorem}
    In PB with cost utilities, \Cref{alg:pb_bw_mes} computes an integral outcome sampled from a lottery that is ex-ante GFS, ex-ante Strong UFS, ex-post \BBOne, and ex-post EJR-x in polynomial time.
\end{theorem}
\begin{proof}
    First, we show that the fractional outcome~$\vec{p}$ returned by \Cref{alg:pb_bw_mes} is feasible, i.e., $\cost(\vec{p}) = B$.
By the design of \Cref{alg:pb_bw_mes}, at any step, we maintain $p_c \leq 1$ for all~$c \in C$.
Next, since each voter starts with a budget of $B/n$ and spends the entirety of their budget, by the construction of $\vec{p}$ in \cref{alg:pb_bw_mes:construct-fractional} we have that
\begin{align*}
\sum_{j \in C} p_j \cdot \cost(j)
= \sum_{j \in C} \frac{\sum_{i \in N} y_{ij}}{\cost(j)} \cdot \cost(j)
= \sum_{i \in N} \sum_{j \in C} y_{ij} = \sum_{i \in N} B/n = B.
\end{align*}

The ex-ante fairness properties follow from \Cref{lem:cost_gfs} and \Cref{lem:cost_sufs}.
By \Cref{thm:PB:Gandhi-et-al}, the outcome sampled in \cref{alg:pb_bw_mes:sample} satisfies ex-post BB1 and contains $W_\MES$.
Finally, it follows from \citet{BFL+23a} that $W\supseteq W_\MES$ satisfies EJR-x.
\end{proof}
\end{document}